\documentclass{article}

\usepackage[final]{neurips_2023}

\usepackage[utf8]{inputenc} 
\usepackage[T1]{fontenc}    
\usepackage[breaklinks, colorlinks,linkcolor=red,urlcolor=blue, citecolor=blue]{hyperref}

\usepackage{url}            
\usepackage{booktabs}       
\usepackage{amsfonts}       
\usepackage{nicefrac}       
\usepackage{microtype}      
\usepackage{xcolor}         

\usepackage{graphicx}
\usepackage{wrapfig}
\usepackage{tikz}
\usepackage{subfigure}
\usepackage{float}
\usepackage{caption}
\usepackage{enumerate}
\usepackage{array}
\usepackage{amssymb}
\usepackage[ruled, linesnumbered, boxed]{algorithm2e}
\usepackage{hyperref}
\usepackage{todonotes}

\usepackage{hyperref}
\usepackage{multirow}
\usepackage{makecell}
\usepackage{bbding}
\usepackage{mathrsfs}
\usepackage{mathtools}

\usepackage{chngcntr}
\counterwithin{table}{section}
\counterwithin{figure}{section}

\usepackage[nameinlink]{cleveref}
\Crefname{figure}{Figure}{Figures}
\crefname{figure}{Figure}{Figures}
\crefname{example}{Example}{Example}
\crefname{theorem}{Theorem}{Theorem}
\crefname{corollary}{Corollary}{Corollary}
\crefname{lemma}{Lemma}{Lemma}
\crefname{proposition}{Proposition}{Proposition}
\crefname{assumption}{Assumption}{Assumption}
\crefname{section}{Section}{Section}
\crefname{algorithm}{Algorithm}{Algorithm}

\usepackage{amsthm,thmtools}
\declaretheorem[name=Theorem,numberwithin=section]{theorem}
\declaretheorem[name=Definition,style=definition]{definition}

\declaretheorem[name=Proposition,numberlike=theorem]{proposition}

\declaretheorem[name=Remark,style=definition,numberwithin=section]{remark}

\usepackage{todonotes}

\usepackage{amsmath,amsfonts,bm}

\def\eqref#1{equation~\ref{#1}}

\def\1{\bm{1}}

\def\rmA{{\mathbf{A}}}
\def\rmB{{\mathbf{B}}}

\def\rmE{{\mathbf{E}}}

\def\rmH{{\mathbf{H}}}

\def\rmR{{\mathbf{R}}}

\def\rmT{{\mathbf{T}}}

\def\ve{{\bm{e}}}

\def\vh{{\bm{h}}}

\def\vr{{\bm{r}}}

\def\vt{{\bm{t}}}

\DeclareMathAlphabet{\mathsfit}{\encodingdefault}{\sfdefault}{m}{sl}
\SetMathAlphabet{\mathsfit}{bold}{\encodingdefault}{\sfdefault}{bx}{n}

\newcommand{\R}{\mathbb{R}}

\title{MQuinE: a cure for ``Z-paradox'' in knowledge graph embedding models}

\author{%
  Yang Liu, Huang Fang \\
  Cognitive Computing Lab, Baidu Research\\
  No.10 Xibeiwang East Road, Beijing 100193, China\\
  \texttt{ \{liuyang173, fanghuang\}@baidu.com } \\
  \AND
  Yunfeng Cai, Mingming Sun \\
  Beijing Institute of Mathematical Sciences and Applications \\
  Huairou District, Beijing, 101408, China \\
  \texttt{\{caiyunfeng, sunmingming\}@bimsa.cn} \\
}

\begin{document}

\maketitle

\begin{abstract}
  Knowledge graph embedding (KGE) models achieved state-of-the-art results on many knowledge graph tasks including link prediction and information retrieval. Despite the superior performance of KGE models in practice, we discover a deficiency in the expressiveness of some popular existing KGE models called \emph{Z-paradox}. Motivated by the existence of Z-paradox, we propose a new KGE model called \emph{MQuinE} that does not suffer from Z-paradox while preserves strong expressiveness to model various relation patterns including symmetric/asymmetric, inverse, 1-N/N-1/N-N, and composition relations with theoretical justification. Experiments on real-world knowledge bases indicate that Z-paradox indeed degrades the performance of existing KGE models, and can cause more than 20\% accuracy drop on some challenging test samples. Our experiments further demonstrate that MQuinE can mitigate the negative impact of Z-paradox and outperform existing KGE models by a visible margin on link prediction tasks.
\end{abstract}

\section{Introduction}
Knowledge graphs (KGs) consist of many facts that connect real-world entities (e.g., humans, events, words, etc.) with various relations. Each fact in a knowledge graph is usually represented as a triplet $(h, r, t)$, where $h, t$ are respectively the head and tail entities and $r$ is the relation; the triplet $(h, r, t)$ indicates that the head entity $h$ has the relation $r$ to the tail entity $t$. Due to the prevalence of relational data in practice, KG has a wide range of applications including recommendation systems \citep{zhang2016collaborative, wang2018dkn, ma2019jointly, wang2019kgat}, 
natural language processing (NLP) \citep{sun2018logician,li2019integration}, question answering (QA) \citep{lukovnikov2017neural,huang2019knowledge} and querying \citep{chen2022fuzzy}. 

Embedding-based models \citep{BengioDVJ03,BleiNJ03} have revolutionized certain fields of machine learning including KG in the past two decades. Simply speaking, knowledge graph embedding (KGE) models map each entity and relation into a vector or matrix and calculate the probability of a fact triple through some score functions. KGE models are space and time efficient. More importantly, KGE models such as TransE \citep{bordes2013translating}, RotatE \citep{sun2019rotate}, OTE \citep{tang2019orthogonal}, etc., are quite expressive; it was shown that KGE models if designed carefully, can capture various relation patterns including symmetry/asymmetry, inversion, composition, injective and non-injective relations. Due to the efficiency and expressiveness of KGE models, embedding-based models have achieved state-of-the-art performance on many KG applications and are widely deployed in practice.

Despite the popularity of KGE models on various KG applications. In this work, we discover a bottleneck, termed as ``Z-paradox'', on the expressiveness of some existing KGE models. We give a short description and illustration of Z-paradox in the following paragraph and present its formal definition and some related properties in \Cref{sec:MquinE}.

\begin{figure}[t]
    \centering
    \includegraphics[width=1.0\textwidth]{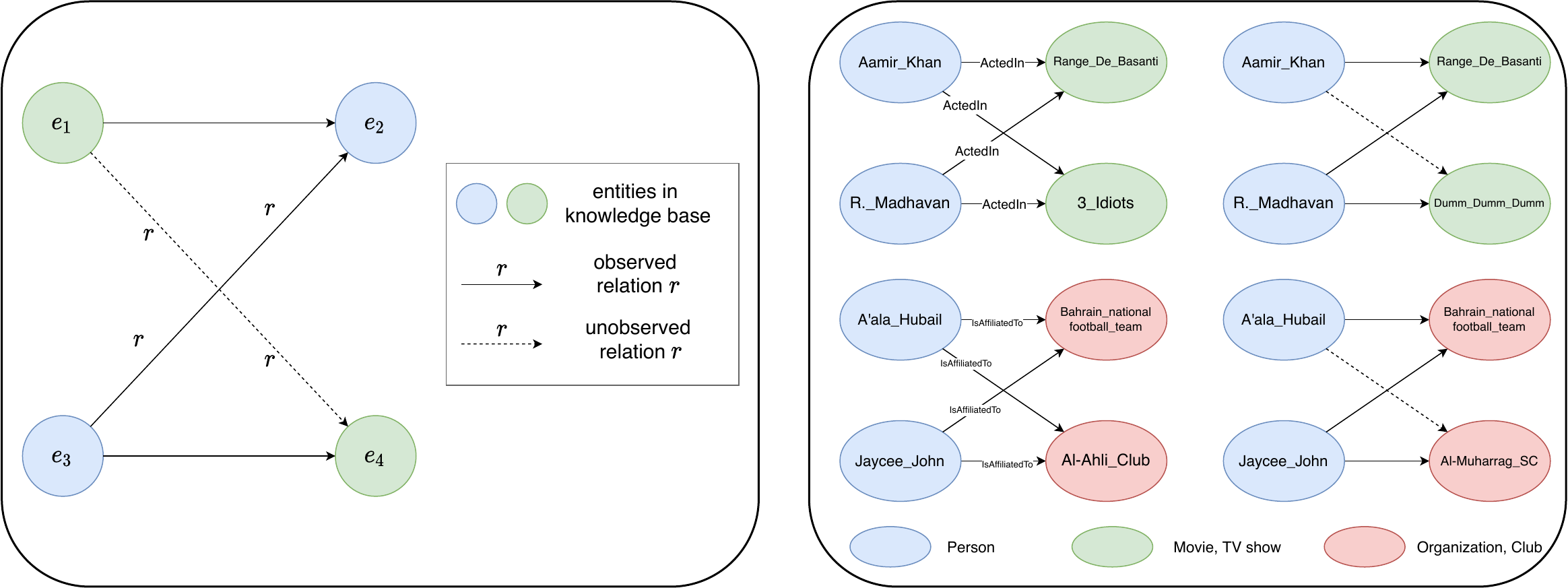}
    \caption{Left panel: an illustration of Z-paradox. Right panel: an illustration of Z-paradox in the YAGO3-10 dataset.}
    \label{fig:example}
    \vspace{-15pt}
\end{figure}


{\bf Z-paradox.}\;
Though popular KGE models (e.g., TransE, RotatE, OTE) have already taken various relation patterns into account, there are still limitations.
In what follows we introduce a limitation of popular KGE models.
Specifically, in \Cref{fig:example}, there are four entities $e_1, e_2, e_3, e_4$, with
$e_1$ linking to $e_2$, $e_3$ linking to both $e_2$ and $e_4$.
The task is to determine whether $e_1$ links to $e_4$ or not. 
A good KGE model should permit both scenarios.
However, we find that many popular KGE models such as TransE \citep{bordes2013translating} and RotatE \citep{sun2019rotate} would guarantee $e_1$ links to $e_4$ regardless of whether $e_1$ actually has relation $r$ to $e_4$ or not. We term this phenomenon \emph{Z-paradox} due to the graph structure in~\Cref{fig:example}. 
To be more concrete, let us take some examples from the YAGO3-10 knowledge base, where some actors and movies are connected by the relation \emph{ActedIn}.
We illustrate the phenomenon via
Figure~\Cref{fig:example},
where the dotted arrow means that 
popular KGE models infer that the arrowtail links to the arrowhead, which contradicts with the true facts.
In the upper left plot,
we observe that \emph{Aamir\_Khan} has acted in \emph{Range\_De\_Basanti}, \emph{R.\_Madhaven} has acted in \emph{Range\_De\_Basanti} and \emph{3\_Idiots}, then the KGE model would infer that \emph{Aamir\_Khan} has acted in \emph{3\_Idiots}, which is correct.
However, in the upper right plot, 
we observe that \emph{Aamir\_Khan} has acted in \emph{Range\_De\_Basanti}, \emph{R.\_Madhaven} has acted in \emph{Range\_De\_Basanti} and \emph{Dumm\_Dumm\_Dumm}, then the KGE model would also infer that \emph{Aamir\_Khan} has acted in \emph{Dumm\_Dumm\_Dumm}, which is incorrect.
The same phenomenon occurs in the lower left and right plots.

We demonstrate that the Z-pattern is indeed a serious issue for standard KG benchmark datasets. For example, about 35\% of the test facts in the FB15k-237 dataset are negatively affected by the Z-pattern, and KGE models such as TransE and RotatE can suffer more than 20\% accuracy drop on these test facts; see~\Cref{exp:Zpatterns} and \Cref{exp:CaseStudy} for details. To mitigate the negative impact of the Z-pattern, we propose a new KGE model to overcome the \emph{Z-paradox}.
Moreover, the new model can ensure both the robust expressiveness and the ability to model various relation patterns, i.e., preserves the good properties of existing KGE models. The new model embeds a triplet $(h,r,t)$ by five matrices $(\rmH, \langle \rmR^h, \rmR^t, \rmR^c \rangle, \rmT)$, where the matrices $\rmH, \rmT$ denote the embeddings of the head entity $h$ and tail entity $t$ respectively, the matrix triplet $\langle \rmR^h, \rmR^t, \rmR^c \rangle$ represents the embedding of the relation $r$.
We term the new model \textbf{M}atrix \textbf{Quin}tuple \textbf{E}mbedding (MQuinE). We show that MQuinE enjoys good theoretical properties and can achieve promising empirical results. Compared with existing KGE models on the challenging FB15k-237 dataset, MQuinE obtains a 10\% improvement of Hit@10 on test facts that are negatively impacted by the Z-pattern, and attains 7\% and 4\% overall improvement of Hit@1 and Hit@10 on all test facts.

Our contributions are summarized as follows:

\noindent{\bf 1)}\; A newly-defined phenomenon in the knowledge graph named \emph{Z-paradox} has been discovered and we prove that existing translation-based KGE models all suffer from Z-paradox. 
Theoretically, we present a necessary condition for the occurrence of Z-paradox. 
    
\noindent{\bf 2)}\; We propose {MQuinE}, a new KGE model that is free from Z-paradox meanwhile can still model complex relations including symmetric/asymmetric, inverse, 1-N/N-1/N-N, and composition relations.

\noindent{\bf 3)}\; Experimental results of {MQuinE} on standard benchmark datasets validate that MQuinE can indeed overcome the negative impact of Z-paradox; MQuinE outperforms existing KGE methods by a large margin on most benchmark datasets.

{\bf Organization.}\; In \Cref{sec:related}, we give an overview of KGE models. 
We introduce the \emph{Z-paradox} bottleneck along with its theoretical properties, and propose our new KGE model {MQuinE} in \Cref{sec:MquinE} and \Cref{sec:methods-2} respectively. In \Cref{sec:exp}, we evaluate the effect of Z-paradox on standard KG benchmarks and empirically compare MQuinE with other competitive baselines. 

{\bf Notation.}\;
We denote the set of entities as $\mathcal{E}$ and the set of relations as $\mathcal{R}$.
Following the conventional notation, we represent a knowledge graph as a set of triplets $\mathcal{O} = \{ (h_i, r_i, t_i) \mid h_i, t_i \in \mathcal{E}, r_i \in \mathcal{R} \}_{i=1}^n$, where $n$ is the number of observed facts. For each entity $e$ and relation $r$, we use their bold version $\ve$ and $\vr$ to denote their embeddings. A KGE model is associated with a score function $s( \cdot ) : \mathcal{E} \times \mathcal{R} \times \mathcal{E} \to \mathbb{R}$. Given a fact $(h, r, t)$, the KGE model tends to predict it to be true if $s( \vh, \vr, \vt )$ is small and false otherwise. We use bold capital letters. e.g., $\rmA, \rmB, \rmH, \rmT$ to denote matrices and use use $\| \cdot \|$ to denote the Euclidean norm of vectors or the Frobenius norm of matrices.

\section{Related works}\label{sec:related}
KGE models received substantial attention and has advanced significantly in the past two decades. Learning graph embeddings that can capture complex relation patterns is the core objective of KGE and many models have been proposed. We summarize some popular KGE models in Appendix (\Cref{table:KGEmethods}). We go through some existing KGE methods and discuss how they relate to our work.


{\bf Translation distance based methods.}\; Translation distance based approaches describe relations assess the plausibility of fact triples by comparing the distances between head's and tail's embeddings following some relation transformations. Inspired by \emph{word2vec} \citep{mikolov2013distributed, mikolov2013efficient}, \cite{bordes2013translating} first introduced the idea of translation invariance into the knowledge graph embedding domain and proposed the TransE model. \citet{sun2019rotate} proposed RotatE and characterized relations as rotations between the head and tail entities in complex space; it was shown that many desirable properties, such as symmetry/asymmetry, inversion, and Abelian composition, can be achieved by RotatE. Later on, OTE \citep{tang2019orthogonal}, DensE \citep{lu2020dense}, HopfE \citep{bastos2021hopfe} are proposed to better model more complex relation patterns. 
Recently, \citet{yu2021mquade} proposed a unified KGE model called MQuadE that can model symmetry/asymmetry, inversion, 1-N/N-1/N-N, and composition relations simultaneously. MQuadE serves as the backbone for the construction of our MQuinE, and is closely related to our work.

{\bf Bilinear semantic matching methods.}\; \citep{nickel2011three} first introduced the idea of tensor decomposition to model triple-relational data.
\citep{yang2014embedding} later proposed a simple and effective bilinear model called DisMult and achieved promising empirical results. Subsequent works such as ComplEX \citep{trouillon2016complex}, TuckER \citep{balavzevic2019tucker}, DihEdral \citep{xu2019relation}, QuatE \citep{zhang2019quaternion} and SEEK \citep{xu2020seek} adopted more complicated bilinear operations to either improve the expressiveness of DisMult or decrease the model complexity.


{\bf Deep learning methods.}\; \citet{vashishth2019composition} proposed COMPGCN to incorporate multi-relational information into graph convolutional networks which leverages a variety of composition operations from knowledge graph embedding techniques to embed both nodes and relations in a graph jointly.
\citep{dettmers2018convolutional} proposed ConvE and used convolutional neural networks to model multi-relational data. 
Subsequent works \citep{nathani2019learning, vashishth2019composition} brought more advanced neural network architectures such as graph convolutional networks and graph attention networks.
More recently, 
\citet{wang2021mixed} proposed M$^2$GNN and embeds entities and relations into the mixed-curvature space with trainable heterogeneous curvatures. 
\citet{zhou2022jointe} proposed JointE and adopted both 1-dimensional and 2-dimensional convolution operations to capture the latent knowledge more carefully. \citet{zhu2021neural}  proposed NBFNet which is a general graph neural network framework that achieves state-of-the-art performance on link prediction tasks. We note that NBFNet is a framework that makes use KGE model as its message function instead of a standalone KGE model. Therefore, we exclude NBFNet from our experiments since the focus of this work is on KGE models. NBFNet can also potentially benefit from using MQuinE as its message function.










\section{Z-paradox and its cure: MQuinE}\label{sec:MquinE}

In this section, we give a formal definition of Z-paradox and propose a new KGE model called MQuinE that can circumvent Z-paradox while having strong expressiveness.

\subsection{Z-paradox}

\begin{definition}[Z-paradox]\label{def:z-paradox}
    Given a KGE model parameterized by $\{ \ve_i \}_{i=1}^{|\mathcal{E}|}, \{ \vr_i \}_{i=1}^{|\mathcal{R}|}$ and a score function $s(\cdot)$ such that $s^* \coloneqq \inf s$.
    For any $ e_1, e_2, e_3, e_4 \in \mathcal{E}, r \in \mathcal{R}$, if 
    \begin{equation}
        s( \ve_1 ,\vr ,\ve_2 ) = s( \ve_3 , \vr ,\ve_2 ) = s( \ve_3 ,\vr ,\ve_4 ) = s^*   \label{eq:links} 
    \end{equation}
    implies that $s( \ve_1 ,\vr ,\ve_4) = s^*$ must hold, then we say the KGE model suffers from Z-paradox.
\end{definition}


Consider a KGE model that suffers from Z-paradox.
If $e_1 \to e_2, e_3 \to e_2, e_3 \to e_4$, i.e., \eqref{eq:links} holds, then $e_1 \to e_4$ must holds regardless the fact is true or not.
It is obvious that a KGE model that suffers from Z-paradox has an inherent deficiency in its expressiveness, and a good KGE model should be able to circumvent Z-paradox. Next, we show that a wide range of existing KGE models indeed suffer from Z-paradox and therefore have limited expressiveness.


\begin{proposition}\label{prop:distanceBasedModels}
    Given a KGE model parameterized by $\{ \ve_i \}_{i=1}^{|\mathcal{E}|}, \{ \vr_i \}_{i=1}^{|\mathcal{R}|}$ and a score function $s(\cdot)$. If $s( \vh, \vr, \vt )$ can be expressed as $\| f( \vh, \vr ) - g( \vt, \vr ) \|$ for some functions $f(\cdot)$ and $g(\cdot)$, and $s^* \coloneqq \inf s = 0$, then the KGE model suffers from Z-paradox.
\end{proposition}

\begin{proof}
    First, we notice that $s( \vh, \vr, \vt ) = s^* = 0$ implies $f( \vh, \vr ) = g(\vt, \vr)$.
    For four entities $e_1,e_2,e_3,e_4$ satisfying $s(\ve_1,\vr,\ve_2) = s(\ve_3,\vr,\ve_2)=s(\ve_3,\vr,\ve_4)=0$, we have
    \begin{equation*}
    f( \ve_1, \vr ) = g( \ve_2, \vr ),\quad
    f(\ve_3, \vr) = g(\ve_2, \vr),\quad
    f(\ve_3, \vr) = g(\ve_4, \vr).
    \end{equation*}
    Then it follows that 
    \begin{align*}
    s(\ve_1, \vr, \ve_4) &= \| f( \ve_1, \vr ) - g( \ve_4, \vr ) \|\\ 
     &= \| [f(\ve_1, \vr) - g(\ve_2, \vr)]  -[f(\ve_3, \vr)
     - g(\ve_2, \vr)] +[f(\ve_3, \vr) - g(\ve_4, \vr)] \|
     =0,
    \end{align*}
        i.e., $(e_1,r,e_4)$ holds. This completes the proof.
    \end{proof}

\begin{remark}
    \Cref{prop:distanceBasedModels} indicates that all existing translation-based KGE models suffer from Z-paradox, including TransE \citep{bordes2013translating}, RotatE \citep{sun2019rotate}, OTE \citep{tang2019orthogonal} and MQuadE \citep{yu2021mquade}, etc.
\end{remark}

Besides translation-based KGE models, KGE models with bilinear score functions may also suffer from Z-paradox under certain conditions.

\begin{remark}
Bilinear score functions can also be transformed into distance-based score functions. For example, we have
\begin{equation*}
2 \mathbf{h}^T \rmR \mathbf{t} = - \| \rmR \mathbf{t} - \mathbf{h} \|^2 + \| \mathbf{h} \|^2 + \| \rmR \mathbf{t} \|^2.
\end{equation*}
When $\mathbf{h},\mathbf{t}$ both have fixed norms and $\rmR$ is a orthogonal matrix, then maximizing $\mathbf{h}^T \rmR \mathbf{t}$ is equivalent to minimizing $\|  \rmR \mathbf{t}- \mathbf{h}\|^2$. Setting $s(h,r,t) = \|\rmR \mathbf{t}- \mathbf{h}\|^2$, by~\Cref{prop:distanceBasedModels}, the model suffers from Z-paradox provided that $\inf s = 0$. The above procedure can be applied to other bilinear KGE models including DisMult \citep{yang2014embedding}, ComplexEX \citep{trouillon2016complex}, DihEdral \citep{xu2019relation}, QuatE \citep{zhang2019quaternion}, SEEK \citep{xu2020seek}, Tucker \citep{balavzevic2019tucker}.
\end{remark}

\begin{table*}[t]
\small
\begin{center}
\begin{tabular}{ccccccc}
\toprule
Model & Sym/Asym & Inversion & Composition & Injective & Non-injective & Z-paradox \\
\midrule
TransE &\XSolidBrush / \CheckmarkBold & \CheckmarkBold & \CheckmarkBold & \CheckmarkBold & \XSolidBrush & \XSolidBrush\\
TransX &\CheckmarkBold / \CheckmarkBold & \XSolidBrush & \XSolidBrush & \CheckmarkBold & \XSolidBrush & \XSolidBrush\\
DisMult &\CheckmarkBold / \XSolidBrush & \XSolidBrush & \XSolidBrush & \CheckmarkBold & \XSolidBrush & \CheckmarkBold\\
ComplEX &\CheckmarkBold / \CheckmarkBold & \CheckmarkBold & \XSolidBrush & \CheckmarkBold & \XSolidBrush & \CheckmarkBold\\
RotatE &\CheckmarkBold / \CheckmarkBold & \CheckmarkBold & \CheckmarkBold & \CheckmarkBold & \XSolidBrush & \XSolidBrush\\
OTE & \CheckmarkBold / \CheckmarkBold & \CheckmarkBold & \CheckmarkBold & \CheckmarkBold & \CheckmarkBold & \XSolidBrush\\
BoxE & \CheckmarkBold / \CheckmarkBold & \CheckmarkBold & \XSolidBrush & \CheckmarkBold & \CheckmarkBold & \CheckmarkBold\\

ExpressivE & \CheckmarkBold / \CheckmarkBold & \CheckmarkBold & \CheckmarkBold & \CheckmarkBold & \CheckmarkBold & \XSolidBrush\\

MQuadE & \CheckmarkBold / \CheckmarkBold & \CheckmarkBold & \CheckmarkBold & \CheckmarkBold & \CheckmarkBold & \XSolidBrush\\
\midrule
\textbf{MQuinE} & \CheckmarkBold / \CheckmarkBold & \CheckmarkBold & \CheckmarkBold & \CheckmarkBold & \CheckmarkBold & \CheckmarkBold \\
\bottomrule
\end{tabular}
\end{center}
\vskip 0.10in
\caption{{The pattern modeling and inference abilities of several models.}}
\label{table:compare}
\end{table*}


\subsection{MQuinE}\label{subsec:relation}
In this section, we introduce our model -- MQuinE;
before that, we review some fundamental relation patterns which need to be considered for KGE models.

{\bf Symmetric/Asymmetric relation.}\; A relation $r$ is symmetric iff the fact triple $(h,r,t)$ holds $\Leftrightarrow$ the fact triple $(t,r,h)$ holds. And a relation $r$ is  asymmetric iff the fact triples $(h,r,t)$ and $(t,r,h)$ do not hold simultaneously.

{\bf Inverse relation.}\; A relation $r_2$ is the inversion of the relation $r_1$ iff the fact triple $(h,r_1,t)$ holds $\Leftrightarrow$ the fact triple $(t,r_2,h)$ holds.

{\bf Relation composition.}\; A relation $r_3$ is the composition of relation $r_1$ and $r_2$ (denoted by $r_3 = r_1 \oplus r_2$) iff the facts $(a,r_1,b)$ and $(b,r_2,c)$ imply the fact $(a,r_3,c)$.

{\bf Abelian (non-Abelian).}\; If $r_1 \oplus r_2 = r_2 \oplus r_1$, the composition $r_1 \oplus r_2$ is Abelian; otherwise, it is non-Abelian.

{\bf 1-N/N-1/N-N relation.}\; A relation $r$ is a 1-N / N-1 relation if there  exist at least two distinct tail/head entities such that $(h,r,t_1),(h,r,t_2)$ / $(h_1,r,t),(h_2,r,t)$ hold.
A relation $r$ is an N-N relation if it is both 1-N and N-1.

Next, we propose MQuinE,
which preserves the aforementioned relation patterns, moreover, circumvents the Z-paradox.
Specifically, for a fact triple $(h,r,t)$, we use
\begin{equation}\label{score}
    s(h,r,t) = \|  \rmH\rmR^h-\rmR^t\rmT + \rmH{\rmR^c}\rmT   \|_F^2
\end{equation}
to measure the plausibility of the fact triple.
Here $\rmH, \rmT\in \R^{d\times d}$ are symmetric matrices and
denote the embeddings of the head entity and tail entity, respectively, and the matrix triple $\langle \rmR^h, \rmR^t, \rmR^c \rangle \in \R^{d\times d}\times \R^{d\times d}\times \R^{d\times d}$ denotes the embedding of the relation $r$.
Specifically, for a true fact triple $(h,r,t)$, we  expect $s(h,r,t) \approx 0$, 
moreover, we hope the score of a false triple to be relatively large.

\subsection{Expressiveness of MQuinE}\label{subsec:method}


In this section, we theoretically show that {MQuinE} is able to model 
symmetric/asymmetric, inverse, 1-N/N-1/N-N, Abelian/non-Abelian compositions relations, more importantly, MQuinE does not suffer from Z-paradox. 

\begin{theorem} \label{thm:sym}
    {MQuinE} can model the symmetry/asymmetry, inverse, 1-N/N-1/N-N relations.
\end{theorem}

\begin{theorem}[Composition] \label{thm:compose}
    {MQuinE} can model the Abelian/non-Abelian compositions of relations.
\end{theorem}

The proofs of \Cref{thm:sym} and \Cref{thm:compose} are provided in Appendix.

\begin{theorem}[No Z-paradox] \label{thm:noZparadox}
    {MQuinE} does not suffer from Z-paradox.
\end{theorem}

\begin{proof}
We show the result via two examples. Let
\begin{align*}
\rmR^h = 
    \begin{bmatrix}
	 1 & 0 \\
	0 & 0 \\
	 \end{bmatrix}, 
  \rmR^t = 
    \begin{bmatrix}
	 0 & 0 \\
	0 & -1 \\
	 \end{bmatrix},
  \rmR^c = 
    \begin{bmatrix}
	 -1 & 0 \\
	0 & -1 \\
	 \end{bmatrix},
  \rmE_1 = 
    \begin{bmatrix}
	 1 & 0 \\
	0 & 1 \\
	 \end{bmatrix}, 
  \rmE_2 = 
    \begin{bmatrix}
	 1 & 0 \\
	0 & 0 \\
	 \end{bmatrix},
  \rmE_3 = 
    \begin{bmatrix}
	 0 & 0 \\
	0 & 1 \\
	 \end{bmatrix}.
\end{align*}

On one hand, set $\rmE_4 = 
    \begin{bmatrix}
	 1 & 1 \\
	1 & 0 \\
	 \end{bmatrix}$, it holds that
\begin{align*}
    s(e_1,r,e_2) = 0,\quad
    s(e_3,r,e_2) = 0, \quad
    s(e_3,r,e_4) = 0,\quad
    s(e_1,r,e_4) = 1. 
\end{align*}
In other words, given a Z-pattern, $e_1$ does not link to $e_4$.
On the other hand, set $\rmE_4 = 
    \begin{bmatrix}
	 1 & 0 \\
    0& -1 \\
	 \end{bmatrix}$, it holds that
\begin{align*}
    s(e_1,r,e_2) = 0,\quad
    s(e_3,r,e_2) = 0, \quad
    s(e_3,r,e_4) = 0,\quad
    s(e_1,r,e_4) = 0. 
\end{align*}
In other words, given a Z-pattern, $e_1$ may link to $e_4$.
This completes the proof.
\end{proof}


\begin{remark}
Set $\rmR^c=0$, MQuinE becomes MQuadE~\citep{yu2021mquade}.
So it is not surprising to draw the conclusion that MQuinE can model various relation patterns (\Cref{thm:sym} and \Cref{thm:compose}) since MQuadE can. The cross term $\rmH\rmR^c\rmT$ plays the central role in circumventing Z-paradox. 
Adding such a cross term to MQuadE to obtain MQuinE is nontrivial, the key insight is 
 \Cref{prop:distanceBasedModels} --- without a cross term between the head and tail entities, a distance-based model must suffer from Z-paradox. 
\end{remark}

The comparison of MQuinE and some existing KGE models in terms of expressiveness are given in \Cref{table:compare}. To our knowledge, MQuinE is the only KGE model that does not suffer from Z-paradox while preserving the ability to capture all relation patterns.

\section{Learning of MQuinE}\label{sec:methods-2}


We introduce the learning of our model in this section, including parameterization, regularization, sampling strategy, and objective function.






\paragraph{Parameterization and regularization.}
We constrain the entity embedding matrices to be symmetric
and parameterize each entity matrix $\mathbf{E}$ by a lower triangular matrix $\mathbf{A}$ and its transpose, i.e.,$\mathbf{E} = \mathbf{A} + \mathbf{A}^{T}$. 
We use the Frobenius norm of entity and relation embedding matrices as regularization.

        

\begin{algorithm}[t]
   \caption{Z-sampling}
   \label{alg:Zsampling}
   \KwIn{the set of entities $\mathcal{E}$, the set of relations $\mathcal{R}$, the set of observed triplets $\mathcal{O}$, a triplet $(h, r, t)$, number of negative samples $m \in \mathbb{N}$, number of Z-samples $k \in \mathbb{N}$.}
   Fix $h, r$, sample $m$ tail entities $\{ t_i \}_{i=1}^m$ s.t. $(h,r,t_i) \notin \mathcal{O}~\forall i \in [m]$, set $S_{\mathrm{neg}} = \{ (h, r, t_i) \}_{i=1}^m$\;
   Set $S_Z = \emptyset$\;
   \For{ $i=1,2,\ldots,m$ }{
        Collect all Z-patterns associated with $(h,r,t_i)$, i.e., $e_2, e_3 \in \mathcal{E}$ s.t. $( h, r, e_2 ), (e_3, r, e_2), (e_3, r, t) \in \mathcal{O}$\;
        $S_Z = S_Z \cup \{ ( h, r, e_2 ), (e_3, r, e_2), (e_3, r, t) \}$;
   }
   Uniform randomly select $k$ triplets in $S_Z$ and remove other triplets from $S_Z$\;
   \KwOut{$S_{\mathrm{neg}}$ and $S_Z$.}
\end{algorithm}


\paragraph{Z-sampling.} 
The negative sampling technique plays an important role in the training of KGE models. Classic negative sampling first sample a batch of observed facts. Then for each fact $(h,r,t)$ in the batch, we fix the head entity $h$ and relation $r$ and sample $m$ tail entities $\{ t_i \}_{i=1}^m$ such that $\{(h,r,t_i) \}_{i=1}^m$ are not observed. Lastly, we perform a gradient step to decrease the score of positive samples and increase the score of negative samples.
To mitigate the effect of Z-patterns more explicitly and fully exploit the benefit of MQuinE, we propose a new sampling technique called \emph{Z-sampling}. Given a positive fact $(h,r,t)$, Z-sampling first sample $m$ negative samples $\{ (h,r,t_i) \}_{i=1}^m$ following exactly the same procedure as the classic negative sampling, then it collects all Z-patterns from observed facts that are related to the sampled negative samples, i.e., 
\begin{align*}
        S_{\mathrm{Z}} = \cup_{i=1}^m \{ (h,r,e_2), (e_3, r, e_2), (e_3, r, t_i) | (h,r,e_2), (e_3, r, e_2), (e_3, r, t_i) \in \mathcal{O} \}.
\end{align*}
Lastly, the Z-sampling uniform randomly samples $k$ facts from $S_Z$ and treats them as positive facts. The detailed algorithm of Z-sampling is given in \Cref{alg:Zsampling}. Z-sampling explicitly tries to minimize the score of positive facts in the Z-pattern and maximize the score of negative facts simultaneously. Experiments on KG benchmark datasets demonstrate the effectiveness of Z-sampling; see \Cref{exp:ablation} for details. Note that~\Cref{alg:Zsampling} can also be applied to ranking the head entity $h$ (fixing $r, t$) with minor modifications, we omit the details.


\paragraph{Objective function.} 
The loss function w.r.t an observed fact $(h,r,t)$ is
\begin{align*}
    \begin{split}
            \mathcal{L}_{ (h,r,t) } = &  -\log \sigma(\gamma-s(h,r,t))  - \lambda_{\mathrm{neg}} \mathbb{E}_{ (h,r,t') \sim S_{\mathrm{neg}} } \left[ \log\sigma(s(h,r,t^{\prime})-\gamma) \right] \\
            & - \lambda_{Z} \mathbb{E}_{ (h,r,t'') \in S_{Z} } \left[ \log \sigma(\gamma-s(h,r,t'')) \right],
    \end{split}
\end{align*}
where $\gamma > 0$ is a pre-defined margin, $\sigma$ is the sigmoid function, i.e., $\sigma(x) = 1/(1+e^{-x})$, $S_{\mathrm{neg}}$ and $S_Z$ are the negative samples and Z-samples as defined in \Cref{alg:Zsampling}, and $\lambda_{\mathrm{neg}}, \lambda_{Z}$ are positive real numbers that control the trade-off between positive and negative samples.








\section{Experiments}\label{sec:exp}

We conduct experiments to demonstrate the impact of Z-paradox for existing KGE models on KG benchmark datasets and evaluate the performance of MQuinE. First, we introduce the experimental setup in \Cref{exp: setup}, including the description of benchmark datasets, evaluation tasks, evaluation metrics and baseline methods.
In \Cref{exp:Zpatterns}, we propose a metric called \emph{Z-value} to quantify the number of Z-patterns related to a given fact and gather statistics of Z-values to quantify the effect of Z-patterns for our experimental datasets. For each dataset, we divide their test samples into easy, neutral, and hard cases according to the Z-value of test facts. In \Cref{exp:CaseStudy}, we evaluate the performance of MQuinE against other competitive baseline methods on the easy, neutral, and hard cases respectively; we also report the overall improvement of MQuinE in \Cref{exp: link prediction}. Lastly, we evaluate the effectiveness of Z-sampling with different KGE models in \Cref{exp:ablation}.


\subsection{Experimental setup}\label{exp: setup}

{\bf Dataset.} We conduct experiments on five large-scale benchmark datasets --- FB15k-237 \citep{toutanova2015observed}, WN18 \citep{bordes2013translating}, WN18RR \citep{dettmers2018convolutional}, YAGO3-10 \citep{mahdisoltani2014yago3}, and CoDEx \citep{safavi2020codex} (CoDEx-L, CoDEx-M, CoDEx-S).
The detailed statistics of these datasets are given in \Cref{appendix: dataset details} (\Cref{table:datasets}).
These datasets contain various relations including 1-N, N-1, N-N, and composition relations, and are suitable for evaluating complex KGE models. 
We follow the train/validation/test split from \cite{sun2019rotate} and divide the observed facts into training, validation, and testing sets by an 8:1:1 ratio. 







{\bf Evaluation task and metrics.} We evaluate the performance of KGE models on the link prediction task. Given a query fact $(h,r,t)$, the link prediction task requires one to fix $h,r$ and rank $t$ among all possible tail entities $t' \in \mathcal{E}$ except those $t'$'s such that $(h,r,t')$'s appear in the training set. We use the mean reciprocal rank (MRR), mean rank (MR), Hits@N (N = 1,3,10) as our evaluation metrics.


{\bf Baselines.} We compare MQuinE with KGE baselines including TransE \citep{bordes2013translating}, RotatE \citep{sun2019rotate}, DisMult \citep{yang2014embedding}, ComplEX \citep{trouillon2016complex}, DihEdral \citep{xu2019relation}, QuatE \citep{zhang2019quaternion}, TuckER \citep{balavzevic2019tucker}, ConvE \citep{dettmers2018convolutional}, OTE \citep{tang2019orthogonal}, BoxE \citep{abboud2020boxe}, HAKE \citep{zhang2020learning}, MQuadE \citep{yu2021mquade}, ExpressivE \citep{pavlovic2023expressive}, DualE \citep{cao2021dual} and HousE \citep{li2022house}.



\subsection{Statistics of Z-patterns}\label{exp:Zpatterns}

\begin{figure}[!t]
\begin{minipage}[t]{\textwidth}
 \begin{minipage}[t]{0.4\textwidth}
 \small
  \centering
       \begin{tabular}[t]{l p{3cm}} 
    \toprule
        Case name & Condition  
     \\ \midrule
        Easy case & $\mathrm{rank}_Z( h, r ,t ) < 10$.
        \\ 
        Neutral case & $n_Z( h,r,t )$ is tied for the $10^{\rm th}$ place.
        \\
        Hard case & otherwise. 
        \\ 
     \bottomrule
     \end{tabular}
\makeatletter\def\@captype{table}\makeatother\caption{ \small{Case splitting description.} }\label{table:Case}
  \end{minipage}
  \quad
  \begin{minipage}[t]{0.5\textwidth}
  \small
   \centering
         \begin{tabular}[t]{lccc}        
           \toprule
    Dataset & Easy case & Neutral case & Hard case  \\
    \midrule
    FB15k-237 & 6,681 (33\%) & 6,546 (32\%) & 7,239 (35\%) \\
    WN18 & 356 (7\%) & 4,331 (87\%) & 313 (6\%) \\
    WN18RR & 314 (10\%) & 2,679 (86\%) & 123 (4\%) \\
    YAGO3-10 & 1,248 (25\%) & 1,074 (21\%) & 2,678 (54\%)\\
        \bottomrule
      \end{tabular}
\makeatletter\def\@captype{table}\makeatother\caption{ \small{Statistics of Z-patterns in the testing set responding to the training set.} }\label{tab: Statistics}
   \end{minipage}
\end{minipage}
\vspace{-10pt}
\end{figure}

We define two statistics that characterize the Z-patterns of a fact. The first one is Z-value, which can be used to measure the number of Z-patterns associated with a fact $(h,r,t)$. The second one is the rank of a fact based on its Z-value.

\begin{definition}[Z-value]
    Given a fact $(h,r,t)$, we define
    \begin{align*}
         n_Z( h, r, t ) := \big| \{  (e_2, e_3) ~|~ e_2 \neq e_3;(h,r,e_2), (e_3, r, e_2), (e_3, r, t) \in \mathcal{O} \} \big|,
    \end{align*}
    which is the number of Z-patterns connected with $(h,r,t)$\footnote{Z-value here differs from the concept of Z-score in statistics.}.
\end{definition}

\begin{definition}
Define
    \[
        \mathrm{rank}_Z( h,r, t ) = \sum_{\substack{ t' \in \mathcal{E},~ (h,r,t') \notin \mathcal{O} }} \mathbf{1}_{\{ n_Z( h, r, t' ) \geq n_Z( h, r, t ) \}},
    \]
    which is the rank of $n_Z( h, r, t )$ among other unobserved candidate facts. 
\end{definition}

Intuitively, for two facts $(h,r,t)$ and $(h,r,t')$, if $n_Z( h, r, t  ) \gg n_Z( h, r, t' )$, then KGE models that suffer from Z-paradox would incline to assign $(h,r,t)$ a lower score\footnote{A lower score means a more plausible fact.} compared with $(h,r,t')$ and rank $(h, r, t)$ before $( h, r, t' )$. Therefore, for a test fact $(h,r,t)$, if there are many facts $(h, r, t')$ such that $(h, r, t')$'s do not appear in the training set and $n_Z(h, r, t') \gg n_Z(h,r,t)$, then for KGE models that suffer from Z-paradox, this test fact should be hard for them to predict. Motivated by the above logic, we can use $\mathrm{rank}_Z(h,r,t)$ to measure the level of difficulty for predicting $(h,r,t)$.
In \Cref{table:Case}, we divide the test facts into three categories, namely, easy, neutral and hard cases.
For easy cases, we require $n_Z(h,r,t)$ to be top-9;
for neural cases, we require
$n_Z(h,r,t)$ is tied for the 10th place; other cases are categorzied into hard cases.



In \Cref{tab: Statistics}, we summarize the ratio of easy, neutral and hard cases in our experimental datasets. We can observe that both FB15k-237 and YAGO3-10 have a notable number of hard cases while most test facts of WN18 and WN18RR are neutral cases.
 



\begin{table*}[t]
\small
    \centering
    \begin{tabular}{|c|c|c|c|c|c|c|c|c|c|}
         \hline 
         \multirow{2}{*}{{Models}}& \multicolumn{3}{c|}{\text { FB15k-237 }} & \multicolumn{3}{c|}{\text { WN18RR }} & \multicolumn{3}{c|}{\text { YAGO3-10 }} \\
\cline { 2 - 10 } & \text { Easy } &  \text { Neutral } & \text { Hard } & \text { Easy } &  \text { Neutral } & \text { Hard } & \text { Easy } &  \text { Neutral } & \text { Hard }  \\
\hline 
\text {ComplEX} &  67.3\% &   47.3\% &   48.9\% &  96.3\% &   50.5\% &   49.6\%
& 0.678\% & 0.455\% & 0.518\% \\ \hline
\text { DisMult } &  63.2\% &   45.2\% &   37.2\% &  96.1\% &   48.6\% &   48.4\% 
& 0.627\% & 0.356\% & 0.526\%\\ \hline
\text { TransE }  &  64.0\% &   48.0\% &   41.9\% &  90.9\% &   48.2\% &   53.7\% 
& 0.665\% & 0.404\% & 0.665\% \\ \hline
\text { RotatE }  &  67.7\% &   48.1\% &   39.7\% &  84.7\% &   53.4\% &   49.2\% 
& 0.751\% & 0.457\% & 0.691\% \\ \hline
\text { MQuinE }  &  69.7\% &   52.9\% &   \textbf{52.7\%} & 85.2\% &   56.8\% &   \textbf{62.2\%}
& 0.783\% & 0.543\% & \textbf{0.759\%}\\ 
\hline
    \end{tabular}
    \vskip 0.10in
    \caption{ {Hits@10 on easy, neutral, and hard cases of FB15k-237 and WN18RR.} }
    \label{tab:CaseStudy}
    \vspace{-10pt}
\end{table*}



\subsection{Case study of MQuinE}\label{exp:CaseStudy}

We evaluate the performance of MQuinE and other baseline methods on the easy/neutral/hard cases, respectively. The results on FB15k-237 and WNRR18 are shown in  \Cref{tab:CaseStudy}, respectively. We can observe that the Hits@10 on hard cases is significantly lower than the Hits@10 on easy cases when using CompLEX, DisMult, TransE and RotatE; the accuracy on hard cases is about 20\% lower than the accuracy on easy cases. This observation indicates that Z-paradox is indeed a serious issue and can significantly degrade the performance of existing KGE models. The results in \Cref{tab:CaseStudy} also show that MQuinE can significantly improve the performance on hard cases; the Hits@10 obtained by MQuinE is 13.0\% higher than RotatE on FB15k-237. In the meanwhile, MQuinE does not sacrifice accuracy on easy and neutral cases; the Hits@10 of MQuinE is 2.0\% and 4.8\% higher than RotatE on easy and neutral cases, respectively. A similar conclusion can be drawn from \Cref{tab:CaseStudy}, the Hits@10 of MQuinE is about 13.0\% higher than RotatE and 10.6\% higher than TransE on WN18RR.

\subsection{Overall evaluation on link prediction}\label{exp: link prediction}




The overall evaluation results of MQuinE and other KGE baseline methods on FB15k-237, WN18RR, YAGO3-10, and CoDEx are presented in \Cref{table:Overall} and some missing results are provided in \Cref{appendix: missing table}. The metric values of baseline methods are taken directly from their original papers. 
Overall, we observe that MQuinE outperforms all the existing KGE methods with a visible margin in all metrics on FB15k-237, where Hits@1 is 7\% higher than other methods and Hits@10 reaches 58.8\%. 
Similarly, MQuinE also exceeds other baseline methods in most metrics on WN18RR, YAGO3-10, WN18 and CoDEx; see \Cref{table:WN18rr}, \Cref{table:YAGO} and \Cref{table:WN} in \Cref{appendix: missing table}. Also, we evaluate MQuinE for node classification task on CoDEx-S and CoDEx-M, and our results shown in \Cref{table:tc} outperforms the baselines mentioned in \citet{safavi2020codex}.


\subsection{Evaluation of Z-sampling}\label{exp:ablation}

We examine the impact of Z-sampling to our model as while as other baseline methods,
the results are given in \Cref{table: ablation}. To fairly evaluation the effect of Z-sampling, we rerun DisMult, ComplEX, TransE, RotatE and MQuadE with/without Z-sampling based on their original implementation.
We can observe that the Z-sampling strategy can significantly improve the performance of DisMult, ComplEX,
but degrades the performance of RotatE and MQuadE a little bit. This demonstrates that Z-sampling is a useful technique for KGE models that do not suffer from Z-paradox.
Not surprisingly, the Z-sampling strategy improves the performance of MQuinE significantly; Z-sampling improves both Hits@1, Hits@3 and Hits@10 for more than 5\%.
Moreover, during our numerical experiments, we also observe that Z-sampling can stabilize the training of MQuinE and makes MQuinE more robust to the different hyperparameter setups. Overall, the superior empirical performance of MQuinE is indispensable to the Z-sampling technique.

\begin{table*}[t]
\small
    \centering
         \begin{tabular}{lccccccccc} 
        \toprule
        \multirow{2}{*}{\textbf{Models}} 
        & \multicolumn{3}{c}{\textbf{FB15k-237}}
        & \multicolumn{3}{c}{\textbf{WN18RR}}
        & \multicolumn{3}{c}{\textbf{YAGO3-10}}
        \\
        \cmidrule{2-4}
        \cmidrule{5-7}
        \cmidrule{8-10}
        & MRR & Hits@1 & Hits@10
        & MRR & Hits@1 & Hits@10
        & MRR & Hits@1 & Hits@10\\
        \midrule
        DisMult &0.241 &0.155  &0.419 
        &0.443 & 0.403 & 0.534 
        &0.340 &0.240 &0.540
        \\
        
        ComplEX &0.247&0.158 &0.428 
        & 0.472 & 0.432 & 0.550 
        &0.360 &0.260 &0.550 
        \\
        
        DihEdral &0.320&0.230  &0.502 
        &0.486 &0.443 &0.557 
        &0.472 & 0.381 &0.643
        \\
        

        TuckER &0.353&0.260  &0.536 
        &0.470 &0.443 &0.526 
        &0.527 &0.446 &0.676 
        \\
        
        ConvE &0.325&0.237  &0.501 
        &0.430 & 0.400   &0.520 
        &0.520 &0.450 &0.660 
        \\
        \midrule
        TransE &0.294 &-  &0.465 
        &0.466 & 0.422 & 0.555 
        & 0.467 & 0.364 & 0.610 
        \\
        
        RotatE &0.336&0.241  &0.530 
        &0.476 & 0.428 & 0.571 
        &0.495&0.402 &0.670
        \\
        
        
        ExpressivE & 0.333 &0.243 &0.512 & 0.482 & 0.407 & \textbf{0.619} & -& -& -\\

        BoxE &0.337 &0.238  &0.538 
        & 0.451 & 0.400  & 0.541 
        & \textbf{0.567} & \textbf{0.494} &0.699
        \\
        HAKE &0.346 &0.250 &0.542
        &\textbf{0.497} &0.452 &0.582
        &0.545 &0.462 &0.694 \\

        MQuadE &0.356 &0.260 &0.549 
        &0.426 &0.427 &0.564 
        &{0.536}&0.449   &0.689
        \\

        DualE & 0.330 & 0.237 & 0.518 & 0.482 & 0.440  & 0.561 & -& -& - \\

        HousE & 0.361 & 0.266 & 0.551 
        & {0.496} & 0.452 & 0.585 
        & {0.565} & 0.487 & 0.703
        \\
        \midrule
        \textbf{MQuinE} & \textbf{0.420} & \textbf{0.332} & \textbf{0.588} &{0.492}&\textbf{0.454}&{0.603} 
        &{0.566}&{0.492}&\textbf{0.711}
        \\
        \bottomrule
        \\
        \toprule
        \multirow{2}{*}{\textbf{Models}} 
        & \multicolumn{3}{c}{\textbf{CoDEx-S}}
        & \multicolumn{3}{c}{\textbf{CoDEx-M}}
        & \multicolumn{3}{c}{\textbf{CoDEx-L}}
        \\
        \cmidrule{2-4}
        \cmidrule{5-7}
        \cmidrule{8-10}
        & MRR & Hits@1 & Hits@10
        & MRR & Hits@1 & Hits@10
        & MRR & Hits@1 & Hits@10
        \\
        \midrule
        RESCAL &0.404 &0.293 &0.623 
        &0.317 &0.244 &0.456
        &0.304 &0.242 &0.419\\
        TransE &0.354 &0.219 &0.634
        &0.303 &0.223 &0.454
        &0.187 &0.116 &0.317\\
        ComplEx &\textbf{0.465} &0.372 &0.646
        &\textbf{0.337} &0.262 &\textbf{0.476}
        &0.294 &0.237 &0.400\\
        ConvE &0.444 &0.343 &0.635
        &0.318 &0.239 &0.464
        &0.303 &0.240 &0.420\\
        TuckER &0.444 &0.339 &0.638
        &0.328 &0.259 &0.458
        &0.309 &0.244 &0.430\\
        \midrule
        \textbf{MQuinE} & {0.443} & \textbf{0.379} & \textbf{0.652}
        & 0.335 & \textbf{0.320} & \textbf{0.476}
        & \textbf{0.326} & \textbf{0.267} & \textbf{0.440}
        \\
        \bottomrule
         \end{tabular}
         \vskip 0.10in
    \caption{ \small{Overall evaluation results on the FB15k-237, WN18RR, and YAGO3-10, and CoDEx datasets.} } \label{table:Overall}
    \vspace{-10pt}
    \end{table*}

\begin{table*}[t]
\small
    \centering
         \begin{tabular}{lccc|ccc} 
        \toprule
        & \multicolumn{3}{c}{\textbf{Without Z-sampling}}
        & \multicolumn{3}{c}{\textbf{With Z-sampling}}\\
        \midrule
        \multirow{2}{*}{\textbf{Models}}
        & \multicolumn{3}{c}{\textbf{Hits@N}} & \multicolumn{3}{c}{\textbf{Hits@N}} \\
          & \textbf{1} & \textbf{3} & \textbf{10} & \textbf{1} & \textbf{3} & \textbf{10} \\
        \midrule
        DisMult &0.155&0.263&0.419    &0.212 ($\uparrow 5.7\%$)& 0.326 ($\uparrow6.3\%$) &0.422 ($\uparrow 0.3\%$) \\
        ComplEX &0.158&0.275&0.428 & 0.231 ($\uparrow7.4\%$) & 0.350 ($\uparrow 7.5\%$) & 0.454 ($\uparrow 2.6\%$) \\
        \midrule
        TransE & - & - &0.465 &0.234  &0.370 &0.484 ($\uparrow 1.9\%$) \\
        RotatE &0.234 &0.366 &0.524
        &0.230 ($\downarrow 0.4\%$)  &0.356 ($\downarrow 1.0\%$)  &0.508 ($\downarrow 1.6 \%$) \\
        MQuadE &0.248 &0.377 &0.529
        &0.269 ($\uparrow 2.1\%$)& 0.342 ($\downarrow 3.5\%$) &0.504 ($\downarrow 2.5\%$) \\
        \midrule
        \textbf{MQuinE} & 0.274 & 0.375 & 0.532 
        & \textbf{0.332 ($\uparrow$ 5.8\%)} 
        & \textbf{0.440 ($\uparrow$ 6.5\%)} 
        & \textbf{0.588 ($\uparrow$ 5.6\%)} \\
        \bottomrule
         \end{tabular}
         \vskip 0.10in
    \caption{ \small{Effect of Z-sampling on the FB15k-237 dataset.}} \label{table: ablation}
    \vspace{-10pt}
    \end{table*}

\section{Conclusion}\label{sec:con}

In this paper, we introduce a phenomenon called Z-paradox and show that many existing KGE models suffer from it both theoretically and empirically.
To overcome the Z-paradox, we propose a new KGE model called \emph{MQuinE} that can circumvent Z-paradox while maintaining strong expressiveness.
Experiments on real-world knowledge bases suggests that the Z-paradox indeed degrades the performance of existing KGE models and strongly support the effectiveness of MQuinE; MQuinE outperforms most existing KGE models by a large margin on standard knowledge completion benchmark datasets.

\bibliography{neurips}
\bibliographystyle{neurips}

\appendix
\section*{Appendix}


The appendix mainly consists five parts.
\begin{enumerate}
    \item Missing proofs to the theorems about model ability in \Cref{subsec:method}.
    \item Descriptions of benchmark datasets used in our experiments.
    \item Missing summary tables for different KGE methods.
    \item Missing results of link prediction.
    \item Implementation details.
\end{enumerate}

\section{Proofs for \Cref{thm:sym} and \Cref{thm:compose}}


For better illustration, we restate
\Cref{thm:sym} as \Cref{thm:symmetry}, \Cref{thm:inverse}, \Cref{thm:N-N}.

\begin{theorem} \label{thm:symmetry}
    \textbf{MQuinE} can model the symmetry/asymmetry relations.
\end{theorem}

\begin{proof}
 A relation $r$ is symmetric iff $(h,r,t) \Leftrightarrow (t,r,h)$. In MQuinE, it requires
    \begin{align*}
        \rmH\rmR^h-\rmR^t\rmT+\rmH\rmR^c\rmT=0 \Leftrightarrow 
        \rmT\rmR^h-\rmR^t\rmH+\rmT\rmR^c\rmH=0.
    \end{align*}
Note that 
    \begin{align*}
        &\rmT\rmR^h-\rmR^t\rmH+\rmT\rmR^c\rmH=0 \\ \Leftrightarrow~&
        -\rmH^{T}(\rmR^t)^{T}+(\rmR^h)^{T} \rmT^{T}+\rmH^{T}(\rmR^c)^{T}\rmT^{T}=0 \\
\stackrel{(a)}{\Leftrightarrow}~&
        -\rmH (\rmR^t)^{T}+(\rmR^h)^{T}  \rmT  +\rmH (\rmR^c)^{T}\rmT =0,
    \end{align*}

where (a) uses the fact that the entity matrix in MQuinE is symmetric.
Hence, if 
    \begin{align*}
         (\rmR^t)^{T} = \mp\rmR^h,\quad
        (\rmR^c)^{T}=\pm\rmR^c,
    \end{align*}
    the relation is symmetric, otherwise, asymmetric.
\end{proof}

\begin{theorem} \label{thm:inverse}
    \textbf{MQuinE} can model the inverse relations.
\end{theorem}

\begin{proof}
    A relation $r_1$ is the inverse relation of  $r_2$ iif  $(h, r_1, t) \Leftrightarrow (t, r_2, h)$. In MQuinE, it requires
    \begin{align*}
        &\rmH\mathbf{R}^h_1 - \rmR^t_1\rmT + \rmH \rmR^c_1 \rmT = 0\Leftrightarrow
        \rmT\rmR^h_2 - \rmR^t_2\rmH + \rmT \rmR^c_2 \rmH = 0. 
    \end{align*}
Using the fact that the entity matrix is symmetric,
we have
\begin{equation*}
\rmT\rmR^h_2 - \rmR^t_2\rmH + \rmT \rmR^c_2 \rmH = 0
\Leftrightarrow
\rmH (\rmR^t_2)^{T} - (\rmR^h_2)^{T}\rmT -  \rmT (\rmR^c_2)^{T} \rmH = 0.
\end{equation*}
Therefore,    
simply set
    \begin{align*}
        \rmR^h_1 = (\rmR^t_2)^{T},\quad
        \rmR^t_1= (\rmR^h_2)^{T}, \quad
        \rmR^c_1 = -(\rmR^c_2)^{T},
    \end{align*}
    the conclusion follows.
\end{proof}

\begin{theorem} \label{thm:N-N}
    \textbf{MQuinE} can model the 1-N/N-1/N-N relations.
\end{theorem}

\begin{proof}
    \textbf{1-N relations.} A relation $r$ is 1-N iff there exist two distinct fact triples $(h, r, t_1)$ and $(h, r, t_2)$. 
    Set $s(h,r,t_1)=s(h,r,t_2)=0$,
we get
    \begin{align*}
        \rmH\rmR^h - \rmR^t\mathbf{T_1} + \rmH \rmR^c\mathbf{T_1} = 0, \\
        \rmH\rmR^h - \rmR^t\mathbf{T_2} + \rmH \rmR^c\mathbf{T_2} = 0,
    \end{align*}
    where $\rmH, \mathbf{T_1}, \mathbf{T_2}$ are the embedding matrices of $h, t_1, t_2$, and $\langle \rmR^h, \rmR^t, \rmR^c \rangle$ is the matrix triple of the relation $r$. 

By simple calculations, we have 
    \begin{equation*}
        (-\rmR^t + \rmH \rmR^c)(\mathbf{T_1} - \mathbf{T_2}) = 0.
    \end{equation*}
Now let us set $\operatorname{rank}(\rmR^t)+\operatorname{rank}(\rmR^c)<d$, then 
\[\operatorname{rank}(-\rmR^t + \rmH \rmR^c)\le \operatorname{rank}(\rmR^t)+\operatorname{rank}(\rmR^c)<d,
\]
i.e., $-\rmR^t + \rmH \rmR^c$ is low rank.
Then $\rmT_1$ and $\rmT_2$ can be distinct,
 \textbf{MQuinE} can model 1-N relations.

    \textbf{N-1 relations.} 
    A relation $r$ is N-1 iff there exist two distinct fact triples $(h_1, r, t)$ and $(h_2, r, t)$. Similar to the proof for 1-N relations, we have
    \begin{align*}\label{equ:1}
        \mathbf{H_1}\rmR^h - \rmR^t\rmT + \mathbf{H_1}\rmR^c\rmT = 0, \\
        \mathbf{H_2}\rmR^h - \rmR^t\rmT + \mathbf{H_2}\rmR^c\rmT = 0,
    \end{align*}
    where $\mathbf{H_1}, \mathbf{H_2}, \rmT$ are the embedding matrices of $h_1, h_2, t$.
Then it follows that
    \begin{equation*}
        (\mathbf{H_1} - \mathbf{H_2})(\rmR^h + \rmR^c\rmT) = 0.
    \end{equation*}
Set $\operatorname{rank}(\rmR^h) + \operatorname{rank}(\rmR^c)<d$, 
 \textbf{MQuinE} can model N-1 relations.

    \textbf{N-N relations.}
    Set $\operatorname{rank}(\rmR^t) + \operatorname{rank}(\rmR^c)<d$ and $\operatorname{rank}(\rmR^h) + \operatorname{rank}(\rmR^c)<d$. The conclusion follows.
\end{proof}

\begin{table}[t]
\small
\caption{Statistics of the datasets.}
\label{table:datasets}
\centering
\begin{tabular}{lrrrrr}
\toprule
Dataset & \#Entity & \# Relation & \# Train & \# Valid & \# Test\\
\midrule
FB15k & 14,951 & 1,345 & 483,142 & 50,000 & 59,071\\
FB15k-237 & 14,541 & 237 & 272,115 &17,535 & 20,466\\
WN18 &  40,943 & 18 & 141,442 & 5,000 & 5,000 \\
WN18RR &  40,943  & 11 & 86,835 & 3,034 & 3,134\\
YAGO3-10 & 123,182 & 37 & 1,079,040 &5,000 &5,000\\
CoDEX-L & 77,951 & 69 & 551,193 & 30,622 & 30,622\\
CoDEX-M & 17,050 & 51 & 185,584 & 10,310 & 10,310\\
CoDEX-S & 2,034 & 42 & 32,888 & 1,827 & 1,828 \\
\bottomrule
\end{tabular}
\vspace{-8pt}
\end{table}

\noindent{\bf Theorem 3.2}\; (Compositions). \;
{\textbf{MQuinE} can model the Abelian/non-Abelian compositions of relations.}

\begin{proof}
A relation $r_3$ is a composition of $r_1$ and $r_2$ (denote as $r_3 = r_1 \oplus r_2$) 
iff we have $(e_1,r_1,e_2),(e_2,r_2,e_3) \rightarrow (e_1,r_3,e_3)$. In \textbf{MQuinE}, it requires
    \begin{align*}
        \mathbf{E_1} \rmR^h_1 - \rmR^t_1\mathbf{E_2} + \mathbf{E_1} \rmR^c_1 \mathbf{E_2} = 0, \\
        \mathbf{E_2} \rmR^h_2 - \rmR^t_2\mathbf{E_3} + \mathbf{E_2} \rmR^c_2 \mathbf{E_3} = 0.
    \end{align*}
Rewrite the above two equalities as
    \begin{align*}
        \mathbf{E_1} \rmR^h_1 - 
        (\rmR^t_1- \mathbf{E_1} \rmR^c_1) \mathbf{E_2} = 0, \\
        \mathbf{E_2} (\rmR^h_2 +\rmR^c_2 \mathbf{E_3}) - \rmR^t_2\mathbf{E_3}  = 0.
    \end{align*}
Then it follows that
\begin{align*}
\mathbf{E_1} \rmR^h_1 (\rmR^h_2 +\rmR^c_2 \mathbf{E_3}) =
        (\rmR^t_1- \mathbf{E_1} \rmR^c_1) \mathbf{E_2}(\rmR^h_2 +\rmR^c_2 \mathbf{E_3})
        =(\rmR^t_1- \mathbf{E_1} \rmR^c_1)\rmR^t_2\mathbf{E_3},
\end{align*}
where is equivalent to
\begin{equation*}
\mathbf{E_1} \rmR^h_1 \rmR^h_2  -\rmR^t_1\rmR^t_2\mathbf{E_3}+
\mathbf{E_1}(\rmR^h_1\rmR^c_2 +
 \rmR^c_1\rmR^t_2)\mathbf{E_3}=0.
\end{equation*}

Define $r_3 = r_1 \oplus r_2$ as
    \begin{align*}
        \rmR^h_3 = \rmR^h_1\rmR^h_2, \quad
        \rmR^t_3 = \rmR^t_1\rmR^t_2, \quad
        \rmR^c_3 = \rmR^h_1\rmR^c_2 + \rmR^c_1\rmR^t_2.
    \end{align*}
 Then we know that $(e_1,r_3,e_3)$ holds.

    \textbf{Abelian/Non-Abelian compositions} By definition, if
    \begin{align*}
        \rmR^h_1\rmR^h_2 = \rmR^h_2\rmR^h_1, \quad
        \rmR^t_1\rmR^t_2= \rmR^t_2\rmR^t_1,\quad
\rmR^h_1\rmR^c_2 + \rmR^c_1\rmR^t_2 =
\rmR^h_2\rmR^c_1 + \rmR^c_2\rmR^t_1,
    \end{align*}
   $r_3$ is an Abelian composition,
    otherwise, non-Abelian.
\end{proof}

\section{Dataset details}\label{appendix: dataset details}

\begin{table}[t]
\small
\caption{ The various relations which can be modeled with different matrices.} \label{table:property}
    \centering
         \begin{tabular}{ll} 
         \toprule
        Relation & Relation matrix property \\
        \midrule
        1-N relation &    $\operatorname{rank}(\rmR^t) + \operatorname{rank}(\rmR^c)<d$ \\
        N-1 relation & $\operatorname{rank}(\rmR^h) + \operatorname{rank}(\rmR^c)<d$ \\
        N-N relation &  $\operatorname{rank}(\rmR^t) + \operatorname{rank}(\rmR^c)<d$ \mbox{and} $\operatorname{rank}(\rmR^h) + \operatorname{rank}(\rmR^c)<d$\\
        Symmetric relation & $(\rmR^t)^{T} = \mp \rmR^h$, $ (\rmR^c)^{T} = \pm \rmR^c$\\
        $r_2$ inversion of $r_1$ & 
        $\rmR^h_1 = (\rmR^t_2)^{T}, \rmR^t_1= (\rmR^h_2)^{T}, \rmR^c_1 = -(\rmR^c_2)^{T}$\\
        Compositions $r_3 = r_1 \oplus r_2$ & 
        $\rmR^h_3 = \rmR^h_1\rmR^h_2, 
        \rmR^t_3 = \rmR^t_1\rmR^t_2, 
        \rmR^c_3 = \rmR^h_1\rmR^c_2 + \rmR^c_1\rmR^t_2$
        \\
        $ r_1 \oplus r_2 $ Abelian composition & 
        $\rmR^h_1\rmR^h_2 = \rmR^h_2\rmR^h_1, 
        \rmR^t_1\rmR^t_2= \rmR^t_2\rmR^t_1, 
        \rmR^h_1\rmR^c_2 + \rmR^c_1\rmR^t_2 =
\rmR^h_2\rmR^c_1 + \rmR^c_2\rmR^t_1
        $
        \\
        $ r_1 \oplus r_2 $ Non-Abelian composition &  
        $\rmR^h_1\rmR^h_2 \neq \rmR^h_2\rmR^h_1$, or 
        $\rmR^t_1\rmR^t_2\neq \rmR^t_2\rmR^t_1$, or 
        $\rmR^h_1\rmR^c_2 + \rmR^c_1\rmR^t_2 \neq
\rmR^h_2\rmR^c_1 + \rmR^c_2\rmR^t_1
        $
        \\
        \bottomrule
         \end{tabular}
    \vspace{-8pt}
    \end{table}


Statistics of the benchmark datasets are summarized in \Cref{table:datasets}. We give a brief overview of them in the following:

\begin{table}[t]
    \small
    \caption{The score functions of different KGE models.} \label{table:KGEmethods}
    \centering
         \begin{tabular}{cllc} 
        \toprule
        Model Category & Model & Score function $s(\vh,\vr,\vt)$ & Representation of param    \\
        \midrule
        \multirow{7}{*}{Bilinear}
        & DisMult & $-\left\langle \mathbf{h},\mathbf{r},\mathbf{t}  \right\rangle$ 
        & $\mathbf{h},\mathbf{r},\mathbf{t} \in \mathbb{R}^k$
         \\
        
        & ComplEX & $-\operatorname{Re}( \left\langle \mathbf{h},\mathbf{r},\bar{\mathbf{t}}  \right\rangle$)  
        & $\mathbf{h},\mathbf{r},\mathbf{t}  \in \mathbb{C}^k$
        \\
        & DihEdral & $-\mathbf{h}^{T}\mathbf{R}\mathbf{t}$ 
        & $\mathbf{h},\mathbf{t} \in \mathbb{R}^{2k}, \mathbf{R} \in \mathbb{D}^k_K$
        \\
        & QuatE  &  $- \left\langle \mathbf{h} \otimes \frac{ \mathbf{r} }{\| \mathbf{r} \|},\mathbf{t} \right\rangle$  
        & $\mathbf{h},\mathbf{r},\mathbf{t} \in \mathbb{H}^k $
        \\
        & SEEK  & $\sum_{x,y} \left\langle \mathbf{r}_x,\mathbf{h},\mathbf{t}_{w_{x,y}} \right\rangle   $
        & $\mathbf{h},\mathbf{r},\mathbf{t} \in \mathbb{R}^k $
        \\
        & TuckER  &  $-\mathcal{W}\times_1 \mathbf{h} \times_2 \mathbf{r} \times_3 \mathbf{t}$ 
        & $\mathbf{h},\mathbf{t} \in \mathbb{R}^k, \mathbf{r}\in \mathbb{R}^l, \mathcal{W}\in \mathbb{R}^{k \times l \times k}$
        \\
        \midrule
        Deep learning
        & ConvE  & $g( \operatorname{vec}(g( [\mathbf{h},\mathbf{r}]*\mathbf{w} ))W )\mathbf{t}$  
        & $\mathbf{h},\mathbf{r},\mathbf{t} \in \mathbb{R}^k, \mathbf{w}\in \mathbb{R}^{m_1}, W \in \mathbb{R}^{m_2} $
        \\
        \midrule
        \multirow{6}{*}{Translation-based}
        & TransE & $\| \mathbf{h}+\mathbf{r}-\mathbf{t} \|$ 
        & $\mathbf{h},\mathbf{r},\mathbf{t} \in \mathbb{R}^k$
         \\
        & RotatE & $\| \mathbf{h} \circ \mathbf{r}-\mathbf{t} \|$ 
        & $\mathbf{h},\mathbf{r},\mathbf{t} \in \mathbb{C}^k, \| \mathbf{r}_i \|=1$
         \\
        & OTE & $ \|  \mathbf{h} \Phi(\mathbf{R})-\mathbf{t}    \| $ 
        & \thead{$\mathbf{h},\mathbf{t} \in \mathbb{R}^k, \mathbf{R} = diag\{\mathbf{R}_1,\cdots,\mathbf{R}_s\}$ \\ $\mathbf{R}_i \in \mathbb{R}^{k/s \times k/s}$}
         \\
        & MQuadE & $\|  \rmH\mathbf{R}-\widehat{\mathbf{R}}\rmT      \|_F  $ 
        & \thead{$\mathbf{H},\mathbf{T},\mathbf{R},\mathbf{\widehat{R}} \in \mathbb{R}^{p \times p}$ \\ $\mathbf{H}, \mathbf{T} \text{ are symmetric }$}
        \\  
        \midrule
        Ours & \textbf{MQuinE} & $\|  \rmH\rmR^h-\rmR^t\rmT + \rmH{\rmR^c}\rmT      \|_F $ 
        & \thead{$\rmH,\rmT,\rmR^h, \rmR^t, \rmR^c \in \mathbb{R}^{p \times p}$ \\ $\rmH, \rmT \text{ are symmetric }$}
         \\
        \bottomrule
         \end{tabular}
    \vspace{-8pt}
    \end{table}

\paragraph{FB15k-237.} FB15k-237 is a subset of the Freebase \citep{bollacker2008freebase} knowledge graph which contains 237 relations. The FB15k \citep{bordes2013translating} dataset, a subset of Free-base, was used to build the dataset by \citep{toutanova2015observed} in order to study the combined embedding of text and knowledge networks. FB15k-237 is more challenging than the FB15k dataset because FB15k-237 strips out the inverse relations.

\paragraph{WN18.}  WN18 is a subset of the WordNet \citep{wordnet}, a lexical database for the English language that groups synonymous words into synsets. WN18 contains relations between words such as \emph{hypernym} and \emph{similar\_to}.

\paragraph{WN18RR.} WN18RR is a subset of WN18 that removes symmetry/asymmetry and inverse relations to resolve the test set leakage problem. WN18RR is suitable for the examination of relation composition modeling ability.

\paragraph{YAGO3-10.} YAGO3-10 is a subset of the YAGO knowledge base \citep{mahdisoltani2014yago3} whose entities have at least 10 relations. The dataset contains descriptive relations between persons, movies, places, etc.

\paragraph{CoDEx.} CoDEx \citep{safavi2020codex} is a set of knowledge graph completion datasets extracted from Wiki- data and Wikipedia that improve upon existing knowledge graph completion benchmarks in scope and level of difficulty.

\section{Missing summary table of score functions and properties for knowledge graph embedding models.}\label{appendix: Summary}

We provide a summary \Cref{table:KGEmethods} of score functions and their mathematical forms for different KGE models.

\section{More experimental results}\label{appendix: missing table}

The results of link prediction with the WN18 dataset and some missing results with FB15k-237, WN18RR, and YAGO3-10. And the overall results with CoDEx dataset \citep{safavi2020codex} are presented in \Cref{table:Overall}.

\begin{table}[t]
\small
    \centering
         \begin{tabular}{lccccc} 
        \toprule
        & \multicolumn{5}{c}{\textbf{Metrics}} \\
        \cline{2-6}
        \multirow{2}{*}{\textbf{Models}}
        &\multirow{2}{*}{\textbf{MRR}} & \multirow{2}{*}{\textbf{MR}} & \multicolumn{3}{c}{\textbf{Hits@N}} \\
         && & \textbf{1} & \textbf{3} & \textbf{10} \\
        \midrule
        DisMult &0.241&254&0.155&0.263&0.419 \\
        ComplEX &0.247&339&0.158&0.275&0.428 \\
        DihEdral &0.320&-&0.230&0.353&0.502 \\
        QuatE &0.311&176&0.221&0.342&0.495 \\
        TuckER &0.353&{162}&0.260&0.387&0.536 \\
        ConvE &0.325&224&0.237&0.356&0.501 \\
        \midrule
        TransE &0.294 &357 & - & - &0.465 \\
        RotatE &0.336&177&0.241&0.373&0.530 \\
        BoxE &0.337 &163 & 0.238 & 0.374 &0.538 \\
        OTE &0.351&- &0.258 &0.388 &0.537 \\
        MQuadE &0.356 &174 &0.260 &0.392 &0.549 \\
        \midrule
        \textbf{MQuinE} & \textbf{0.420} & \textbf{109} & \textbf{0.332} & \textbf{0.440} & \textbf{0.588} \\
        \bottomrule
         \end{tabular}
         \vskip 0.10in
         \caption{ Overall evaluation results on the FB15k-237 dataset.} \label{table:FB15k-237}
    \vspace{-10pt}
    \end{table}

\begin{table}[t]
\small
    \centering
         \begin{tabular}{lccccc} 
        \toprule
        & \multicolumn{5}{c}{\textbf{Metrics}} \\
        \cline{2-6}
        \multirow{2}{*}{\textbf{Models}}
        &\multirow{2}{*}{\textbf{MRR}} & \multirow{2}{*}{\textbf{MR}} & \multicolumn{3}{c}{\textbf{Hits@N}} \\
         && & \textbf{1} & \textbf{3} & \textbf{10} \\
        \midrule
        DisMult &0.443 & 4999 & 0.403 & 0.453 & 0.534\\
        ComplEX & 0.472 & 5702 & 0.432 & 0.488 & 0.550\\
        DihEdral &0.486&-&0.443&0.505&0.557 \\
        TuckER &0.470&-&0.443&0.482&0.526 \\
        ConvE &0.430 & - & 0.400 &0.440  &0.520 \\
        \midrule
        TransE &0.466&-&0.422&-&0.555 \\
        RotatE &0.476&{3340}&{0.428}&0.492&0.571 \\
        BoxE & 0.451 & 3207 & 0.400 & 0.472 & 0.541 \\
        MQuadE &0.426&6114&0.427&0.462&0.564 \\
        \midrule
        \textbf{MQuinE}&\textbf{0.492}&\textbf{2599}&\textbf{0.454}&\textbf{0.518}&\textbf{0.603} \\
        \bottomrule
         \end{tabular}
         \vskip 0.10in
         \caption{ Overall evaluation results on the WN18RR dataset.} \label{table:WN18rr}
\vspace{-10pt}
\end{table}

\begin{table}[t]
\small
    \centering
         \begin{tabular}{lccccc} 
        \toprule
        & \multicolumn{5}{c}{\textbf{Metrics}} \\
        \cline{2-6}
        \multirow{2}{*}{\textbf{Models}}
        &\multirow{2}{*}{\textbf{MRR}} & 
        \multirow{2}{*}{\textbf{MR}} & 
        \multicolumn{3}{c}{\textbf{Hits@N}} \\
         & & & \textbf{1} & \textbf{3} & \textbf{10} \\
        \midrule
        DisMult &0.340 &5926 &0.240&0.380&0.540 \\
        ComplEX &0.360 &6351 &0.260&0.400&0.550 \\
        DihEdral &0.472 & - &0.381&0.523&0.643 \\
        TuckER &0.527&3306&0.446&0.576&0.676 \\
        ConvE &0.520&2792&{0.450}&0.560&0.660 \\
        \midrule
        RotatE &0.495&1767 &0.402&0.550&0.670 \\
        BoxE & \textbf{0.567} & 1164 &\textbf{0.494} &0.611&0.699 \\
        MQuadE &{0.536}&1337  &0.449 &0.582 &0.689 \\
        \midrule
        \textbf{MQuinE} 
&{0.566}&\textbf{992}&{0.492}&\textbf{0.629}&\textbf{0.711}
        \\
        \bottomrule
         \end{tabular}
         \vskip 0.10in
         \caption{ Overall evaluation results on the YAGO3-10 dataset.} \label{table:YAGO}
         \vspace{-10pt}
    \end{table}

\begin{table}[t]
\small
    \centering
         \begin{tabular}{lccccc} 
        \toprule
        & \multicolumn{5}{c}{\textbf{Metrics}} \\
        \cline{2-6}
        \multirow{2}{*}{\textbf{Models}}
        &\multirow{2}{*}{\textbf{MRR}} & \multirow{2}{*}{\textbf{MR}} & \multicolumn{3}{c}{\textbf{Hits@N}} \\
         && & \textbf{1} & \textbf{3} & \textbf{10} \\
        \midrule
        ComplEX &0.941&-&0.936&0.945&0.947 \\
        DihEdral &0.946&-&0.942&0.949&0.954 \\
        TuckER &0.953&-&\textbf{0.949}&0.955&0.958 \\
        ConvE &0.943&374 &0.935 &0.946 &0.956 \\
        \midrule
        TransE &-&263&-&-&0.754 \\
        RotatE &{0.949}&\textbf{09}&{0.944}&{0.952}&0.959 \\
        MQuadE &0.897&268&0.893&0.926&0.941 \\
        \midrule
        \textbf{MQuinE} & \textbf{0.958} & 189 &{0.937}&\textbf{0.957}&\textbf{0.975} \\
        \bottomrule
         \end{tabular}
         \vskip 0.10in
         \caption{ Results of link prediction on the WN18 dataset.} \label{table:WN}
    \vspace{-10pt}
    \end{table}

\begin{table}[t]
\small
    \centering
         \begin{tabular}{lcccc} 
        \toprule
        \multirow{2}{*}{\textbf{Models}} 
        & \multicolumn{2}{c}{\textbf{CoDEx-S}}
        & \multicolumn{2}{c}{\textbf{CoDEx-M}}
        \\
        \cline{2-5}
        & Acc & F1-score
        & Acc & F1-score
        \\
        \midrule
        RESCAL 
        & 0.843 & 0.852
        & 0.818 & 0.815\\
        TransE 
        & 0.829 & 0.837
        & 0.797 & 0.803
         \\
        ComplEx 
        & 0.836 & 0.846
        & 0.824 & 0.818
         \\
        ConvE 
        & 0.841 & 0.846
        & 0.826 & \textbf{0.829}
         \\
        TuckER 
        & 0.840 & 0.846
        & 0.823 & 0.816
         \\
        \midrule
        \textbf{MQuinE} 
        & \textbf{0.876} & \textbf{0.883}
        & \textbf{0.831} & 0.828
        \\
        \bottomrule
         \end{tabular}
         \vskip 0.10in
         \caption{ {Overall evaluation results on the CoDEx datasets for triple classification.} } \label{table:tc}
    \vspace{-10pt}
    \end{table}


\section{Implementation details}

We introduce the training details of our model and baselines in this section. All experiments are conducted on a machine with 8 3080Ti GPUs with 12GB memory. 

\subsection{Our model}

We present the implementation details on the training of MQuinE used in \Cref{sec:exp}.

We initialized the element of entity matrix from the normal distribution $\mathcal{N}(0,0.01)$ and the relation matrices $\rmR^h,\rmR^t,\rmR^c$ are initiated as the identity matrix. We apply grid search to find the best hyper-parameters of MQuinE. The tuning ranges of hyper-parameters are as follows: the number of Z-samples $k \in \{10,32,64\}$, the dimension of entity and relation matrices $d \in \{20,25,32,35,42\}$, the batch size $b \in \{256,512,1024\}$, the self-adversarial temperature $\alpha \in \{0.5,1\}$, the fixed margin $\gamma \in \{6,9,12,15,21,24\}$, the number of negative samples $m \in \{128,256,512,1024\}$, the initial learning rate $\eta \in \{10^{-3},10^{-4},10^{-5}\}$, the regularization coefficient 
$\lambda_{\mathrm{reg}} \in \{10^{-4},5\times 10^{-4},10^{-3},10^{-2}, 10^{-1}\}$, the negative sampling coefficient $\lambda_{\mathrm{neg}} \in \{0.5,1.0,2.0\}$.

The best hyperparameters on each dataset are given in \Cref{tab:parameters}.

\begin{table}[t]
\small
    \centering
    \begin{tabular}{|l|c|c|c|c|c|c|c|c|}
         \hline 
         Dataset & $k$ & $d$ & $b$ & $\alpha$ &  $\gamma$ & $m$ & $\lambda_{\mathrm{reg}}$ & $\lambda_{\mathrm{neg}}$   \\
\hline 
FB15k-237 & 32 & 38 & 1024 & 0.5 & 12.0 & 256 & 0.01 & 1.0\\ \hline 
WN18 & 64 & 35 & 1024 & 1.0 & 9.0 & 256 & 5e-3 & 1.0\\
\hline 
WN18RR & 64 & 35 & 1024 & 0.5 & 12.0 & 512 & 0.01 & 1.0 \\ \hline 
YAGO3-10 & 64 & 18 & 4096 & 1.0 & 32.0 & 512 & 5e-3 & 1.0\\ \hline
CoDEx-S & 32 & 32 & 1024 & 0.5 & 12.0 & 256 & 0.01 & 1.0\\ \hline
CoDEx-M & 32 & 32 & 1024 & 0.5 & 12.0 & 256 & 0.01 & 1.0\\ \hline
CoDEx-L & 32 & 32 & 1024 & 0.5 & 12.0 & 256 & 5e-3 & 1.0\\
\hline
    \end{tabular}
    \vskip 0.10in
    \caption{The best hyperparameters of MQuinE on four datasets in our experiments. }
    \label{tab:parameters}
    \vspace{-10pt}
\end{table}

\newpage
\subsection{Baselines}
We present the details of our re-implementation of the baseline methods in \Cref{sec:exp}. We follow the implementation of 
RotatE\footnote{\url{https://github.com/DeepGraphLearning/KnowledgeGraphEmbedding}} and use the best configuration of TransE, RotatE, DisMult, and ComplEX reported\footnote{\url{https://github.com/DeepGraphLearning/KnowledgeGraphEmbedding/blob/master/best\_config.sh}}. We set the number of Z-sampling $k$ to 32 for all the datasets in \Cref{table: ablation}.






\end{document}